\documentclass{article}
\usepackage[utf8]{inputenc}
\usepackage{amssymb}
\usepackage{parskip}
\usepackage{amsmath}
\usepackage{relsize}

\usepackage[margin=1.5in]{geometry}

\usepackage{amsthm}
\theoremstyle{plain}

\usepackage{parskip}
\usepackage{xcolor}
\usepackage{indentfirst}
\usepackage{graphicx}

\newtheorem{theorem}{Theorem}
\newtheorem{lemma}{Lemma}
\newtheorem{proposition}{Proposition}

\theoremstyle{remark}
\newtheorem{remark}[]{Remark}

\title{An Upper Bound for the Number of Gravitationally Lensed Images in a Multiplane Point-mass Ensemble}
\author{Sean Perry}
\date{October 31, 2019}

\begin{document}

\maketitle

\begin{abstract}

\indent 
 Herein we prove an upper bound on the number of gravitationally lensed images in a generic multiplane point-mass ensemble with $K$ planes and $g_i$ masses in the $i^{\text{th}}$ plane. With $E_K$ and $O_K$ the sums of the even and odd degree terms respectively of the formal polynomial $\prod_{i=1}^K(1+g_iZ)$, the number of lensed images of a single background point-source is shown to be bounded by $E_K^2+O_K^2$. Previous studies concerning upper bounds for point-mass ensembles have been restricted to two special cases: one point-mass per plane and all point-masses in a single plane. 
\end{abstract}

\section{Introduction}

 When mass exists in the intervening space between an observer and a light source, gravitational effects can create multiple images of a single source. The multiplane ensemble is a simplification from the fully three dimensional model of gravitational potential to a finite collection of planes containing two dimensional potentials. Each of these planes and the plane containing a light source can be thought of as a copy of the complex plane. Light travels physically from the source, through the planes containing mass, to an observer beyond the first plane. Light rays follow null-geodesics in space-time, and as such will travel the same path from observer to source as from source to observer \cite[pg 77]{PetBook}. We consider then a light ray traveling back through its apparent position in the first plane to its source. At each plane, the light ray is bent by the gravitational field corresponding to the mass distribution in that plane, eventually impacting the source plane. The map $\eta:\mathbb{C}\to\mathbb{C}$ which takes the position of an image at $x$ in the first plane to the corresponding source $w$ in the source plane is called the \textit{lensing map}. Correspondingly, if $\eta(x)=w$ then an observer will see an apparent image of a source at $w$ at location $x$ in the first plane. Thus we are interested in solutions to the equation
 
\begin{equation}
\label{eq:etaw}
\eta_w(x):=\eta(x)-w=0,
\end{equation}
 which will be referred to as the \textit{lensing equation}. In Figure \ref{fig:ray}, $x$ is (and $z$ is not) a solution to (\ref{eq:etaw}). Rather than particular solutions, we are interested in the maximum number of images that can be produced by such an ensemble.

\begin{figure}
    \includegraphics[width=\textwidth]{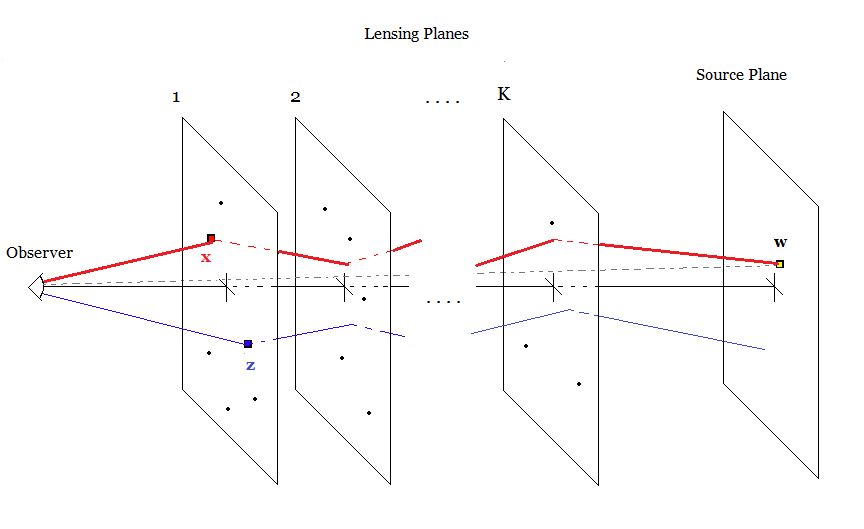}
\caption{Visualization of a K-plane point-mass lensing ensemble. Here, $x$ is (and $z$ is not) a solution to the lensing equation $\eta_w(x)=0$.} 
\label{fig:ray}
\end{figure}

It is possible that for some positions in the first plane a light ray passing through this point may impact one of the point-mass locations in one of the planes. Such rays physically may not reach the source plane. Moreover, as can be seen in (\ref{eq:recursivelens}) below, $\eta$ will fail to be defined at these points. Thus, the domain of $\eta$ is $\mathbb{C}$ with these so-called \textit{obstruction points} removed. Additionally, for a given $w$, the critical points of $\eta_w$, those $x$ for which $\text{det}[\text{Jac}\; \eta_w](x)=0$, generically consists of simple, closed curves. Given a particular mass ensemble,  the set of all critical values of $\eta$, those $w$ for which some $x$ is a critical point of $\eta_w$, are called \textit{caustics}. The images corresponding to caustic source $w$ are called \textit{degenerate images} of the lensing map $\eta$ or the \textit{degenerate zeros} of (\ref{eq:etaw}). Sard's theorem assures us that the caustic of $\eta$ has measure zero \cite{PetWer10}. We assume hereafter that our ensemble is generic in the sense sources do not lie on a caustic ($w$ is not a critical value of $\eta$). 

 The lensing map is formulated based on the physical parameters of the ensemble, in particular to the gravitational potential of the masses in each plane.  It may include terms corresponding to dark matter and continuous distributions of mass. It may also include shear, a bulk skewing effect from large and distant masses, such as galaxies far, far away. These terms can affect the upper bound on the number of solutions; however, for simplicity, we will take these terms to be zero. More information on the general formulation of the lensing map can be found in \cite[Ch. 6]{PetBook}(ch. 6). With these considerations, the lensing map for point-masses has the form:
 
\begin{equation}
\label{eq:recursivelens}
\begin{cases}
&\eta (x) = x - \mathlarger{\sum_{k=1}^K \alpha_k(x_k)},\;\;\; \text{where}  \\
&x_1 = x \;\;, \;\; x_i = x - \mathlarger{\sum_{j=1}^{i-1}\beta_{j,i} \alpha_j(x_j)} \;\;\; \text{for}\;\; i=2,3,...,K+1, \;\;\; \text{and}   \\

&\alpha_i(x_i)=\mathlarger{\sum_{\ell=1}^{g_i}\frac{m_{i,\ell}}{\overline{x}_i-\overline{y}_{i,\ell}}}.
\end{cases}
\end{equation}

 Here, the functions $\alpha_i$ are the bending-angle vectors due to the gravitational potential of the point-masses in the $i^{\text{th}}$ plane. The $x_i$ are values in $\mathbb{C}$ indicating the impact parameter of a light ray on the $i^{\text{th}}$ plane. The value $x$ is $x_1$, the impact parameter in the first plane, and can be treated as the sole independent variable upon which all subsequent $x_{i}$ with $i=2,3,...,K+1$ are dependent. (See equation (\ref{eq:unfurledlens}) below.) The $y_{i,\ell}$ are constant values in $\mathbb{C}$ indicating location of point-masses in the $i^{\text{th}}$ plane with scaled masses given by $m_{i,\ell}$, which are positive real values. The positive integers $g_i$ denote the number of masses in the $i^{\text{th}}$ plane. The $\beta_{j,i}$ are nonzero values in $\mathbb{R}$ representing scaling constants derived from the angular diameter distance between planes. The bar represents complex conjugation. 

It is perhaps useful in order to appreciate the full complexity of the lensing equation to write its non-recursive form, at least for the first several planes: 

\begin{align}
\label{eq:unfurledlens}
\eta_w(x) = 
 x &-w-\sum_{i=1}^{g_1}\frac{m_{1,i}}{\overline{x}-\overline{y}_{1             ,i}}\nonumber \\
    &- \sum_{j=1}^{g_2} \frac{m_{2,j}}{\overline{x} - \overline{y}_{2,j} - \beta_{1,2}\cdot \sum_{i=1}^{g_1}\frac{m_{1,i}}{x - y_{1,i}}}\\
    &- \sum_{k=1}^{g_3}\frac{m_{3,k}}{\overline{x}-\overline{y}_{3,k}-\beta_{1,3}\sum_{i=1}^{g_1}\frac{m_{1,i}}{x-y_{1,i}}-\beta_{2,3}\sum_{j=1}^{g_2}\frac{m_{2,j}}{x - y_{2,j} - \beta_{1,2}\sum_{i=1}^{g_1} \frac{m_{1,i}}{\overline{x}-\overline{y}_{1,i}}}}\nonumber\\
     &- ...\nonumber .
\end{align}

It will also be convenient (for manipulations such as clearing denominators) to write the lensing equation in a modified recursive form:

\begin{equation} 
\label{eq:fullrecurslens}
\begin{cases}
    &\eta_w(x) = x-w + \mathlarger{\sum_{k=1}^K Q_k(x)}, \;\;\;\text{where} \\
&Q_1 = -\mathlarger{\sum_{i=1}^{g_1}\frac{m_{1,i}}{\overline{x}-\overline{y}_{1,i}}}, \;\;\; \text{and}\\

&Q_j = -\mathlarger{\sum_{\ell=1}^{g_j}\frac{m_{j,\ell}}{\overline{x}-\overline{y}_{j,\ell}+\sum_{i=1}^{j-1}\beta_{i,j}\overline{Q}_i}}\;\;\; \text{for} \;\; j=2,3,...,k.
\end{cases}
\end{equation}

 We now state the upper bound for the number of images in a multiplane point-mass ensemble:

\medskip
\begin{theorem}
For a generic multiplane point-mass ensemble with $K$ planes and $g_i$ masses in the $i^\text{th}$ plane, let $F(Z)$ be the formal polynomial defined by $$F(Z)=\prod_{i=1}^K (1+g_i Z).$$ Let $$E_K = \frac{F(1)+F(-1)}{2} \quad \text{and} \quad O_K = \frac{F(1)-F(-1)}{2},$$ the sums of the coefficients of the even and odd degree respectively.Then the number $N$ of images satisfies 

\begin{equation}
\label{eq:NEKOK}
    N \leq E_K^2 + O_K^2.
\end{equation}

\end{theorem}

Previous work on the problem of finding an upper bound on the number of images given a number of lensing masses $g$ has relied on further simplifications and specification in the setup of the problem. Witt in 1990 examined the single plane case \cite{Witt90} and produced an upper bound of $g^2+1$ for ensembles with zero shear. (For positive shear, the upper bound was $(g+1)^2$.) Petters in 1997 examined the case of one mass in each of g planes \cite{Pet97} and produced an upper bound of $2(2^{2(g-1)}-1)$. In 2001, Rhie \cite{Rhie2001} conjectured that the upper bound for the single plane case was in fact $5(g-1)$. In her 2003 paper \cite{Rhie}, she produced a particular point-mass ensemble that produced this number of images. In 2006, Khavinson and Neumann \cite{KN} settled the conjecture in the affirmative that the upper bound for the one plane case could be improved to $5(g-1)$. Upper bounds of more restrictive ensembles have also been obtained. For example, the upper bound of $g+3$ images for collinear masses as shown by Kuznia and Lundberg \cite{Kuznia}.

Concerning continuous distributions of matter, it has been noted \cite[pg 436]{PetBook} that the number of images in a general lensing system can always be increased by adding a continuously distributed elliptical potential at a sufficiently far point. Thus a general upper bound on the number of images does not exist. Work with continuous distributions of matter have produced upper bounds for physically pertinent distributions of mass. For instance, Fassnacht, Keeton, and Khavinson proved in 2007 that an elliptical uniform mass distribution, such as from an elliptical galaxy, produces a maximum of four images located external to the lens \cite{Keeton}. Khavinson and Lundberg 2010 \cite{Khav2010} also showed that there are at most eight external lensed images due to an elliptic lens with isothermal density. This was later improved to six by the work of Bergweiler and Eremenko \cite{Berg2010}. Lower bounds in general have been obtained through the application of Morse theory in gravitational lensing \cite{PetWer10},\cite{PetMorse}. See \cite{Pet10} for further discussion of results related to the image counting problem.

\begin{center}
\textbf{Acknowledgements}
\end{center}

\indent
The author would like to thank the American Mathematical Society and especially Charles Keeton, Arlie Petters, and Marcus Werner for organizing a Math Research Community on the Mathematics of Gravity and Light in the Summer of 2018, for allowing the author to participate, and for their knowledge, insight, and encouragement. The author would most of all like to thank Erik Lundberg for his extensive assistance, and advice.

\section{Methodology}

Following the method of Petters in \cite{Pet97}, \cite{PetWer10}, \cite[pg 457]{PetBook}, we recall that for a generic set of parameters there are only a finite number of solutions to the lensing equation, none of which are degenerate. We clear the nested denominators of (\ref{eq:fullrecurslens}) to produce a polynomial in $x$ and $\overline{x}$ which we will call the ``lensing polynomial". Solutions to the lensing equation will also be mutual zeros of this polynomial and its conjugate. Calculating their degrees in $x$ and $\overline{x}$, we employ the method of resultants to obtain an upper bound on the number of mutual zeros and thus on the number of solutions to the lensing equation. This will lead to the result stated in Theorem 1.

First, we establish some essential properties and definitions for bivariate polynomials. We define the \textit{degree vector} of a polynomial of two variables $\rho(x,y)$ written 
$$\rho(x,y) = \sum_{m=0}^s a_m(y)x^m = \sum_{n=0}^t b_n(x)y^n\\$$
as a 2-tuple:
$$\text{deg}(\rho(x,y)):= (s,t).$$
Thus the degree vector records the highest exponent of each variable. We say that $\rho$'s degree in $x$ and $y$ is $s$ and $t$ respectively, written
$$
\text{deg}_x(\rho)=s \;\;\;\; \text{and} \;\;\;\; \text{deg}_y(\rho)=t.
$$
The degree vector admits a partial order: given another polynomial $\upsilon(x,y)$ such that deg$(\upsilon) = (v,w)$ we write $\text{deg}(\rho) \geq \text{deg}(\upsilon)$ iff $s \geq v$ and $t \geq w$.
\bigskip

\begin{lemma}{\textit{Properties of the degree vector of bivariate polynomials:}}
\label{lem:degvectprop}
    Let $p$ and $q$ be polynomials in $x$ and $y$ given by 
$$p(x,y) = \sum_{i=0}^s a_i(y)x^i = \sum_{i=0}^t b_i(x)y^i$$ and 
$$q(x,y) = \sum_{i=0}^v c_i(y)x^i = \sum_{i=0}^w d_i(x)y^i$$
where $a_i$ and $c_i$ are polynomials in $y$, and $b_i$ and $d_i$ are polynomials in $x$, and thus $\text{deg}(p)=(s,t)$ and $\text{deg}(q)=(v,w)$. Let $n$ be a positive integer. Then 
\begin{align*} 
&\text{a})\;\; & \text{deg}(p+q) &= (\text{max}\{s,v\}, \text{max}\{t,w\})\\
&\text{b})\;\; & \text{deg}(p \cdot q) &= (s+v , t+w)\\
&\text{c})\;\; & \text{deg}(p^n) &= (n \cdot s, n \cdot t)\\ 
&\text{d}) & \text{deg}(q)\leq \text{deg}(p) \implies  \text{deg}(p+q) &= \text{deg}(p)
\end{align*}

\end{lemma}

\begin{remark}
The proof of lemma \ref{lem:degvectprop} follows straight from the basic rules for adding and multiplying polynomials of a single variable. 
\end{remark}

With this notation, we present a version of the resultant theorem summarized from \cite[pg 158-159]{CoxBook}. More information on the method of resultants and the so-called ``Sylvester-matrix'' can be found in \cite[Ch. 3, \textsection 6]{CoxBook}and \cite[pg 437-438]{PetBook}.

\bigskip

\begin{theorem}{The Resultant Theorem}

Let $p(x,y)$ and $q(x,y)$ be coprime polynomials in $x$ and $y$ with coefficients from an algebraically closed field. Then the number $M$ of solutions common to both $p = 0$ and $q = 0$ is bounded by
    $$
    M \leq \text{deg}_x(p)\cdot \text{deg}_y(q) + \text{deg}_y(p) \cdot \text{deg}_x(q).
    $$
\end{theorem}

Observe that $\eta_w$, as defined in (\ref{eq:fullrecurslens}), may be considered either as a function of the single variable $x$ or as a \textit{rational} function of $x$ and $\overline{x}$. 

To ``clear the denominators" in $\eta_w$ we define
\begin{equation}
\label{eq:D}
\left \{ 
\begin{aligned}
&D_1 = \prod_{\ell=1}^{g_1}(\overline{x}-\overline{y}_{1,\ell})\;\;\;\text{and}  \\
&D_j = \prod_{\ell=1}^{g_j}(\overline{x}-\overline{y}_{j,\ell} + \sum_{i=1}^{j-1} \beta_{i,j}\overline{Q}_i)\cdot (\prod_{k=1}^{j-1} \overline{D}_k)^{g_j} 
\end{aligned}
\right.
\end{equation}
 for $i=2,3,...,K$ and write $D = \prod_{i=1}^K D_i$. Note each $D_j$ is a polynomial; even though rational functions appear in its definition, they are immediately cancelled. The product $P:=\eta_w\cdot D$ is also a polynomial which we will call the \textit{lensing polynomial}. The corresponding \textit{conjugate polynomial} is $\overline{P} = \overline{\eta_w} \cdot \overline{D}$ where the conjugation of $\eta_w$ and $D$ entails conjugating every coefficient \textit{and} variable. Note that if $\eta_w(x_0, \overline{x_0})=0$ then $D(x_0, \overline{x_0})\neq 0 \neq \overline{D}(x_0, \overline{x_0})$ since if $D(x_0, \overline{x_0})=0$ or $\overline{D}(x_0, \overline{x_0})=0$ then $\eta_w$ or $\overline{\eta_w}$ (and hence $\eta_w$) are undefined at this point. This is because the zeros of $D$ are precisely the obstruction points of $\eta_w$. As in \cite[pg 463]{PetBook}, any solutions of $\eta_w(x)=0$ will also be a solution of both
\begin{equation}
P(x,\overline{x})=0 \;\;\;\; \text{and} \;\;\;\; \overline{P}(x,\overline{x})=0.
\end{equation}
Thus we can produce an upper bound on the number of lensed images by bounding the number of common zeros of $P$ and $\overline{P}$. While applying the Resultant Theorem, we treat the  variable $\overline{x}$ as a second variable, independent from $x$. We note in passing that this may overestimate the number of mutual solutions of $P$ and $\overline{P}$. 

In order to apply the Resultant Theorem to $P$ and $\overline{P}$, we must first justify that these can be treated as coprime. If the lensing polynomial and its conjugate were not coprime they would have infinitely many common zeros  which lie on the locus $p=0$, where $p$ is the greatest common factor of $P$ and $\overline{P}$. Generically, only a finite number of these will be zeros of $\eta_w$, but it is unknown a priori which lie on the locus $p=0$. This prevents us from cancelling $p$ immediately and working with the relatively prime polynomials $p_1=P/p$ and $p_2=\overline{P}/p$.  However, if $p(x_0,\overline{x_0})=0$ and $\eta_w(x_0,\overline{x_0})=0$ then, as we explain below, the lensing equation is degenerate at $(x_0,\overline{x_0})$. Since generically the lensing equation has no degenerate solutions, we may then safely eliminate $p$ without fear of losing some of the zeros pertaining to lensed images. Writing $\eta_w = p\cdot p_1 / D$ and $\overline{\eta_w} = p\cdot p_2/\overline{D}$, we have 
\begin{align*}
\eta_w 
    &= \left(\frac{\eta_w + \overline{\eta_w}}{2},\frac{\eta_w - \overline{\eta_w}}{2i}\right)\\
    &= \left(p\cdot \frac{1}{2}\left(\frac{p_1}{D}+\frac{p_2}{\overline{D}}\right),p\cdot\frac{1}{2i}\left(\frac{p_1}{D}-\frac{p_2}{\overline{D}}\right)\right)\\
    &= (p\cdot U, p\cdot V) 
\end{align*}
with $U=(p_1/D+p_2/\overline{D})/2$ and $V=(p_1/D-p_2/\overline{D})/2i$.
Momentarily using the natural identification of $\mathbb{C}$ with $\mathbb{R}^2$ with coordinates $(x,y)$ we write
$$
\text{Jac}(\eta_w)=\begin{bmatrix} 
 p \cdot \frac{\partial U}{\partial x} & p \cdot \frac{\partial V}{\partial x} \\[.25cm]
p \cdot \frac{\partial U}{\partial y} & p \cdot \frac{\partial V}{\partial y} 
\end{bmatrix}
+
\begin{bmatrix} 
\frac{\partial p}{\partial x}\cdot U & \frac{\partial p}{\partial x}\cdot V \\[.25cm]
\frac{\partial p}{\partial y}\cdot U & \frac{\partial p}{\partial y}\cdot V
\end{bmatrix}.
$$
 Note that the  first matrix is zero when we plug in $(x_0, \overline{x_0})$ and the determinant of the second matrix is then identically zero, which confirms our assertion. We may thus cancel $p$ and work with the relatively prime reduced polynomials $p_1$ and $p_2$. The degree vectors of $p_1$ and $p_2$ are less than or equal to the degrees of $P$ and $\overline{P}$ respectively. Therefore the resultant polynomial of $p_1$ and $p_2$ in the case of $p$ non-constant will have less than or equal to the degree of the resultant of $P$ and $\overline{P}$ in the case of $p$ constant. Since this only improves the upper bound, we will simply treat the original $P$ and $\overline{P}$ as coprime.

\section{Proof of Theorem 1}
\bigskip

Since multiplication of a bivariate polynomial by constants does not change its degree vector, Lemma \ref{lem:degvectprop} justifies setting each $m_{i,j}$ and $\beta_{i,j}$ in (\ref{eq:fullrecurslens}) and (\ref{eq:D}) equal to one. Since setting the constants to 1 respects degree, we may use this simplified polynomial as a surrogate for the  purposes of calculating degree. We will refer to $n_w$ and $D$ hereafter with the substitutions $m_{i,j}=1=\beta_{i,j}$ in place.

The lensing polynomial $P=\eta_w\cdot D$  includes numerous terms of the form
$$
\sum_{\ell=1}^{g_j}(...\prod_{\substack{j=1 \\ j\neq \ell}}^{g_j}(...)).
$$
The degree vector of such a term is less than a cancellation-free term, that is one of the form 
$$
\sum_{\ell=1}^{g_j}(...\prod_{j=1}^{g_j}(...)).
$$
We now have the following expression for the lensing polynomial
\begin{equation}
\label{eq:P}
P(x,\overline{x}) := (x-w)\cdot \prod_{i=1}^k D_i + \sum_{\ell=1}^k Q_i\cdot D_i\prod_{\substack{i=1 \\ i\neq \ell}}^{g_j} D_i  
\end{equation}
Since we are merely interested in the degree of this polynomial, we can restrict our attention to finding the degree of 
\begin{equation}
\label{eq:Phat}
\hat{P}(x,\overline{x}) := (x-w)\cdot \prod_{i=1}^kD_i
\end{equation}
since $\text{deg}(P) = \text{deg}(\hat{P})$ as shown in Proposition 1 below.
\bigskip

\begin{proposition}
With the above notation and definitions, we have the  following:
\begin{align*}
    &a)\;\text{deg}(P) = \text{deg}(\hat{P})\\ 
    &b)\;\text{deg}(P) = (E_K,O_K)
\end{align*}
\end{proposition}

\begin{proof}
To prove $a)$ it suffices to show that  $\text{deg}(\hat{P})\geq \text{deg}(P)$ since equality then follows from Lemma \ref{lem:degvectprop}. We induct on the  number of planes, $K$. Comparing the terms in $P$ and  $\hat{P}$ in (\ref{eq:Phat}) and (\ref{eq:P}) and invoking Lemma \ref{lem:degvectprop}, it suffices to show that

\begin{equation}
\text{deg}(D_j) \geq \text{deg}(D_j \cdot Q_j) \;\; \text{for all}  \;\; j=1,2,...,K .
\end{equation}

Note: 
\begin{align*}
D_1\cdot Q_1 &= \prod_{\ell=1}^{g_1}(\overline{x}-\overline{y}_{1,\ell}) \cdot \sum_{\ell=1}^{g_1}\frac{1}{\overline{x}-\overline{y}_{1,\ell}} \\
   &= \sum_{\ell=1}^{g_1}  \prod_{\substack{i=1 \\ i\neq \ell}}^{g_1}(\overline{x}-\overline{y}_{1,i}),
\end{align*}
thus $\text{deg}(D_1\cdot Q_1) = (0,(g_1-1))$ and $\text{deg}(D_1) = (0, g_1)$. This establishes the base case for induction on $K$. Assume, then, that $a)$ holds for all $\ell=1,2,...K-1$. 

Note for two polynomials $S(x,\overline{x})$ and $R(x,\overline{x})$ we have
$$
\text{deg}(S)\leq \text{deg}(R) \;\;\; \iff \;\;\; \text{deg}(\overline{S})\leq \text{deg}(\overline{R}) .
$$
It then follows from our induction hypothesis that for $0\leq i \leq K-1$
$$
    \text{deg}(\overline{D}_i\cdot \overline{Q}_i) \leq \text{deg}(\overline{D}_i).
$$
We may write $D_K$ as
\begin{align*}
    D_K &= \prod_{i=1}^{g_K}(\overline{x}-\overline{y}_{K,i}+\sum_{j=1}^{K-1} \overline{Q}_j)\cdot (\prod_{\ell=1}^{K-1}\overline{D}_\ell)^{g_K}\\
    &= \prod_{i=1}^{g_K}\left((\overline{x}-\overline{y}_{K,i})\cdot \prod_{\ell=1}^{K-1}\overline{D}_\ell +\sum_{j=1}^{K-1}\overline{D}_j\cdot \overline{Q}_j \cdot \prod_{\substack{\ell=1 \\ \ell\neq j}}^{K-1}\overline{D}_\ell\right) .
\end{align*}
Note by Lemma 1 and our induction hypothesis we have that
\begin{equation*}
\text{deg}(\prod_{j=1}^{K-1}\overline{D}_j) \geq \text{deg}\left(\sum_{j=1}^{K-1} \overline{Q}_j\cdot \overline{D}_j \cdot \prod_{\substack{i=1 \\ i\neq j}}^{K-1}\overline{D}_i\right),
\end{equation*}
and
\begin{equation*}
\begin{split}
\text{deg}(D_K)&=\text{deg}\left(\prod_{i=1}^{g_K}(\overline{x}-\overline{y}_{K,i})\cdot (\prod_{\ell=1}^{K-1}\overline{D}_\ell)^{g_K}\right).
\end{split}
\end{equation*}

We have also that 
\begin{equation*}
\begin{split}
D_K\cdot Q_K &= \prod_{i=1}^{g_K}(\overline{x}-\overline{y}_{K,i}+\sum_{j=1}^{K-1} \overline{Q}_j) (\prod_{\ell=1}^{K-1}\overline{D}_\ell)^{g_K} \cdot \sum_{i=1}^{g_K} \frac{1}{\overline{x}-\overline{y}_{K,i}+\sum_{j=1}^{K-1}\overline{Q}_j}\\
 &= \sum_{i=1}^{g_K}\left[\prod_{\substack{m=1 \\ m\neq i}}^{g_K}(\overline{x}-\overline{y}_{K,m}+\sum_{j=1}^{K-1} \overline{Q}_j)\right] (\prod_{\ell=1}^{K-1}\overline{D}_\ell)^{g_K} \\
 &= \sum_{i=1}^{g_K}\left[\prod_{\substack{m=1 \\ m \neq i}}^{g_K}\left((\overline{x}-\overline{y}_{K,m}) \prod_{\ell=1}^{K-1}\overline{D}_\ell+(\sum_{j=1}^{K-1}\overline{Q}_j) \prod_{\ell=1}^{K-1}\overline{D}_\ell\right)\right] (\prod_{\ell=1}^{K-1}\overline{D}_\ell)\\
 &=\sum_{i=1}^{g_K}\left[\prod_{\substack{m=1 \\ m \neq i}}^{g_K}\left((\overline{x}-\overline{y}_{K,m}) \prod_{\ell=1}^{K-1}\overline{D}_\ell+(\sum_{j=1}^{K-1}\overline{Q}_j  \overline{D}_j) \prod_{\substack{\ell=1 \\ \ell \neq j}}^{K-1}\overline{D}_\ell\right)\right] (\prod_{\ell=1}^{K-1}\overline{D}_\ell).
\end{split}
\end{equation*}

Again through our induction hypothesis and Lemma 1 we have
\begin{equation*}
\begin{split}
    \text{deg}(D_K \cdot Q_K) &= \text{deg}\left(\left[\sum_{i=1}^{g_K} \left(\prod_{\substack{m=1 \\ m\neq i}}^{g_K}(\overline{x}-\overline{y}_{K,m})\prod_{\ell=1}^{K-1}\overline{D}_\ell\right)\right]\cdot \prod_{\ell=1}^{K-1}\overline{D}_\ell\right)\\
    &\leq \text{deg}\left(\prod_{m=1}^{g_K} (\overline{x}-\overline{y}_{K,m}) \cdot (\prod_{\ell=1}^{K-1}\overline{D_\ell})^{g_K}\right)\\
    &= \text{deg}(D_K).\\
\end{split}
\end{equation*}

 This completes the induction and proves $a)$. 

To prove $b)$, first note that 
\begin{equation}
\label{eq:Pdeg}
\text{deg}(P)=\left(1+\text{deg}_x(\prod_{i=1}^KD_i),\text{deg}_{\overline{x}}(\prod_{i=1}^KD_i)\right).
\end{equation}

To relate how (\ref{eq:Pdeg}) corresponds to (\ref{eq:NEKOK}), examine how the degree changes as we add an additional plane. Let us call $P_K$ the polynomial for $K$ planes and $P_{K-1}$ the polynomial for a matching ensemble, save without the $K^\text{th}$ plane. We must examine the new terms created going from $E_{K-1}$ and $O_{K-1}$ to $E_{K}$ and $O_{K}$ then find this same pattern in the calculation of the degree of $P_{K}$. As noted, the terms of $E_{K-1}$ and $O_{K-1}$ are the coefficients of the terms of even and odd degree respectively in the polynomial $\prod_{i=1}^{K-1} (1+g_iZ)$. Let us consider two polynomials $\prod_{i=1}^{K-1} (1+g_iZ)$ and $\prod_{i=1}^{K} (1+g_iW)$ corresponding to the generation of degree for $K-1$ and $K$ planes respectively. Specifically consider how their coefficients will differ. Expanding and inspecting, the constant term will of course still be one, and the coefficient of $W^n$ where $1< n < K$ will have a coefficient comprised of the previous coefficient of $Z^n$  plus new terms comprised of the coefficients of $Z^{n-1}$ multiplied by $g_K$. This pattern even includes the coefficient of $W^K$ where the coefficient of $Z^K$ is zero. Now examine the degree of $D_K$:
\begin{equation}
\label{eq:degDK}
\begin{split}
\text{deg}(D_K)&=\text{deg}\left(\prod_{i=1}^{g_K}(\overline{x}-\overline{y}_{K,i})\cdot (\prod_{\ell=1}^{K-1}\overline{D}_\ell)^{g_K}\right) \\
&=\left(g_K \cdot \text{deg}_{\overline{x}}(\prod_{i=1}^{K-1}D_i), g_K + g_K \cdot \text{deg}_x (\prod_{i=1}^{K-1}D_i)\right) .
\end{split}
\end{equation}
Hence when we multiply $\prod_{i=1}^{K-1}D_i$ by $D_j$ this increases the degree by the amount indicated in equation (\ref{eq:degDK}). But this and (\ref{eq:Pdeg}) tells us 
$$
\text{deg}(P_K)=\left(\text{deg}_{x}(P_{K-1}) + g_K \cdot \text{deg}_{\overline{x}}(P_{K-1}), \text{deg}_{\overline{x}}(P_{K-1})+g_K\cdot \text{deg}_{x}(P_{K-1})\right)
$$
This is precisely the changes in the coefficients described in the preceding discussion. This proves $b)$ and thus Proposition 1.

\end{proof}

Note that the arithmetic operations used to clear the denominators commute with taking the complex conjugate. Hence
\begin{align*}
    \text{deg}(P(x,\overline{x})) &= (E_k,O_k),  \\
    \text{deg}(\overline{P}(x,\overline{x})) &= (O_k,E_k).
\end{align*}

Theorem 1 then follows from the Resultant Theorem.

\begin{flushright}
$\blacksquare$
\end{flushright}
\bigskip

Figure \ref{fig:chart} clarifies some of the preceding formulas and computations. Looking at the degree of $P$ as we clear the first five planes, the combinatorial pattern described by Theorem 1 emerges. The column sums give the degree of $P$ in terms of $x$ (left column sum) and $\overline{x}$ (right column sum). The horizontal lines demark the contributions from each plane. Note the left column contains all even combinations of the $g_i$, including 1 as the empty combination, and the left column contains all the odd combinations of the $g_i$, up to $i=5$. The $1$ in the upper left corner does not come from the $D_j$ terms, but from the $(x-w)$. (Accordingly, with $\text{deg}(D_1)= (0,g_1)$, we thereby get the ``extra'' $g_K$ in (\ref{eq:degDK}).) Recognizing the entries as the symmetric functions, we confirm the left (right) column sum is the sum of the coefficients of terms of even (odd) degree in the formal polynomial $\prod_{i=1}^5 (1+g_iZ)$.

 \begin{figure}
 \begin{equation}
\left(
\begin{array}{c|c}
    \color{red} 1 & \color{blue}g_1 \\
    \hline
    \color{red}g_1 \cdot g_2 & \color{blue}g_2 \\
    \hline
    \color{red}g_1 \cdot g_3 & \color{blue}g_3 \\
    \color{red}g_2 \cdot g_3 & \color{blue}g_1 \cdot g_2 \cdot g_3 \\
    \hline
    \color{red}g_1 \cdot g_4 & \color{blue}g_4 \\ 
    \color{red}g_2 \cdot g_4 & \color{blue}g_1 \cdot g_2 \cdot g_4 \\
    \color{red}g_3 \cdot g_4 & \color{blue}g_1 \cdot g_3 \cdot g_4 \\
    \color{red}g_1 \cdot g_2 \cdot g_3 \cdot g_4 & \color{blue}g_2 \cdot g_3 \cdot g_4 \\
    \hline
    \color{blue}g_1 \cdot \color{green}g_5 & \color{red}1 \cdot \color{green}g_5 \\
    \color{blue}g_2 \cdot \color{green}g_5 & \color{red}g_1 \cdot g_2 \cdot \color{green}g_5 \\
    \color{blue}g_3 \cdot \color{green}g_5 & \color{red}g_1 \cdot g_3 \cdot \color{green}g_5 \\
    \color{blue}g_4 \cdot \color{green}g_5 & \color{red}g_1 \cdot g_4 \cdot \color{green}g_5 \\
    \color{blue}g_1 \cdot g_2\cdot g_3 \cdot \color{green}g_5 & \color{red}g_2 \cdot g_3 \cdot \color{green}g_5 \\
    \color{blue}g_1 \cdot g_2\cdot g_4 \cdot \color{green}g_5 & \color{red}g_2 \cdot g_4 \cdot \color{green}g_5 \\
    \color{blue}g_1 \cdot g_3\cdot g_4 \cdot \color{green}g_5 & \color{red}g_3 \cdot g_4 \cdot \color{green}g_5 \\  
    \color{blue}g_2 \cdot g_3\cdot g_4 \cdot \color{green}g_5 & \color{red}g_1 \cdot g_2 \cdot g_3 \cdot g_4 \cdot \color{green}g_5 \\  
    \hline
    ... & ...\\
\end{array}
\right)
\end{equation}
\caption{The degree of the lensing polynomial in $x$ and $\overline{x}$. The left column sum is its degree in $x$ and the right column sum is its degree in $\overline{x}$. The horizontal lines demark which plane contributed the terms to the overall degree.}
 \label{fig:chart}
\end{figure}

\section {Concluding Remarks}

 Bezout's Theorem \cite[pg 420]{CoxBook} provides an alternative method to compute an upper bound on the number of lensed images. One could apply this theorem to the lensing equation as defined in (\ref{eq:recursivelens}), treating it as system of $2K$ equations in $2K$ unknowns. In this application, however, Bezout's bound gives $(E_K + O_K)^2$ rather than $E_K^2 + O_K^2$ and is thus strictly greater than the bound we produced. 

 We note that (\ref{eq:NEKOK}) corresponds to special upper bounds calculated in previous work. In the case where $K=1$, we have an upper bound given by $$ N \leq 1+g_1^2$$ and hence recover the upper bound given by \cite{Witt90}. In \cite[pg 464]{PetBook} and \cite{Pet97} the problem is considered with one mass per plane. In this work, just after application of the Resultant Theorem, we have an upper bound given for $g_1 = ... = g_K = 1$ by $$N\leq 2^{2K-1}.$$ This agrees with Theorem 1 for $K>1$ since then$$E_K=O_K= 1+\sum_{\ell=1}^{K-1} 2^{\ell-1} = 2^{K-1},$$ and so $$ E_K^2+O_K^2 = 2(2^{K-1})^2 = 2^{2K-1}.$$ It should be mentioned that Petters goes on to reduce this bound by 2. 
 
 It is interesting to compare the bound given by (\ref{eq:NEKOK}) and the bound given in \cite{Pet97} in the case where we could reasonably consider the masses to be separated into two clusters in two planes $(K=2)$. Take the total number of point-masses to be $g=2n$ for some positive integer $n$. Say there are $g_1$ in the first plane and  $g_2$ in the second plane with $g_1=g_2=n$. Our upper bound for this case would be $1 + 6n^2 + n^4$ and the bound given in \cite{Pet97} would be $ 2^{4n-1}-2$. 
 
Recall that in the single plane case the quadratic bound given by Witt was improved by Khavinson and Neumann to a linear one. Their methodology utilizing complex dynamics may not immediately extend to the multiplane case since the lensing map is then no longer harmonic. It therefore remains to be seen if one can produce a bound which is linear in each of the $g_i$. For example, does the bound $N\leq \prod_{i=1}^K 5(g_i-1)$ hold?   
   
\bibliographystyle{abbrv}
\bibliography{grav}

\end{document}